\documentclass[11pt,reqno]{amsart}
\usepackage{geometry}
\usepackage{amsmath, amsfonts, amssymb,amscd,graphicx}
\usepackage{mathrsfs,url}

\DeclareMathOperator{\Id}{Id}
\DeclareMathOperator{\Forb}{Forb}
\DeclareMathOperator{\Age}{Age}

\newcommand{\fA}{{\bf A}}

\newcommand{\bE}{{\mathfrak E}}
\newcommand{\mQ}{\mathbb Q}

\date{\today}

\author%[M. Bodirsky]
{Manuel Bodirsky}
\address{Institut f\"{u}r Algebra\\TU Dresden\\01062 Dresden\\Germany}
\email{Manuel.Bodirsky@tu-dresden.de}
   \urladdr{http://www.math.tu-dresden.de/~bodirsky/}
    \thanks{The author has received funding from the European Research Council under the European Community's Seventh Framework Programme (FP7/2007-2013 Grant Agreement no. 257039, CSP-Infinity). The final publication is available at Springer via \url{http://dx.doi.org/}.}

\title[Finite Relation Algebras with Normal Representations]{Finite Relation Algebras \\with Normal Representations}

\theoremstyle{plain}

    \newtheorem{thm}{Theorem}
        \newtheorem{theorem}[thm]{Theorem}

    \newtheorem{proposition}[thm]{Proposition}
    \newtheorem{corollary}[thm]{Corollary}
    
    \newtheorem{conjecture}[thm]{Conjecture}

\theoremstyle{definition}

    \newtheorem{definition}{Definition}

    \newtheorem{example}[thm]{Example}

\newcommand{\bA}{\ensuremath{\mathfrak{A}}}
\newcommand{\bB}{\ensuremath{\mathfrak{B}}}
\newcommand{\bC}{\ensuremath{\mathfrak{C}}}

\newcommand{\Fresse}{Fra\"{i}ss\'{e}}

\newcommand{\ignore}[1]{}

\DeclareMathOperator{\Csp}{CSP}

\DeclareMathOperator{\Aut}{Aut}

\DeclareMathOperator{\Pol}{Pol}

\begin{document}

\maketitle
\begin{abstract}
One of the traditional applications
of relation algebras is to provide a setting 
for infinite-domain constraint satisfaction problems.
Complexity classification for these computational problems has been one of the major open research challenges of this application field. The past decade has brought significant progress on the theory of constraint satisfaction, both over finite and infinite domains. This progress has been achieved independently from the relation algebra approach. The present article translates 
the recent findings into the traditional relation algebra setting, 
and points out a series of open problems at the interface between model theory and the theory of relation algebras. 
\end{abstract}

\section{Introduction}
One of the fundamental computational problems for a relation algebra $\fA$ is the \emph{network satisfaction problem for $\fA$}, which is to determine for a given $\fA$-network $N$ whether it is satisfiable in some representation of $\fA$ (for definitions, see Sections~\ref{sect:ra} and~\ref{sect:networks}). 
Robin Hirsch named in 1995 the 
 \emph{Really Big Complexity Problem (RBCP)}
 for relation algebras,  
which is to \emph{`clearly map out which (finite) 
relation algebras are tractable and which are intractable'}~\cite{Hirsch}. 
For example, for the Point Algebra the network satisfaction problem is in P and for Allen's Interval Algebra it is NP-hard. One of the standard methods to show that the network satisfaction problem for a finite relation algebra is in P is via establishing local consistency. 
The question whether the network satisfaction problem for $\fA$ can be solved by local consistency methods is another question that has been studied intensively 
for finite relation algebras $\fA$ (see~\cite{Qualitative-Survey} for a survey on the second question). 

%For many of the classical relation algebras (such %as the Point Algebra or Allen's Interval algebra)

If $\fA$ has a fully universal square representation (we follow the terminology of Hirsch~\cite{Hirsch})
then the network satisfaction problem for $\fA$
can be formulated as a constraint satisfaction problem (CSP) for a countably infinite structure. 
%When $\fA$ has a \emph{regular representation}.
The complexity of constraint satisfaction is a research direction that has seen quite some progress in the past years. The \emph{dichotomy conjecture} of Feder and Vardi from 1993 states that every CSP for a finite structure is in P or NP-hard; the \emph{tractability conjecture}~\cite{JBK} is a stronger conjecture that predicts precisely which CSPs are in P and which are NP-hard. 
 Two independent proofs of these conjectures appeared in 2017~\cite{BulatovFVConjecture,ZhukFVConjecture}, based on concepts and tools from universal algebra. An earlier result of Barto and Kozik~\cite{BoundedWidth} gives an exact characterisation of those finite-domain CSPs that can be solved by local consistency methods. 

Usually, the network satisfaction problem
 for a finite relation algebra $\fA$ cannot be formulated as a CSP for a finite structure. However, suprisingly often 
 it can be formulated as a CSP for a countably infinite \emph{$\omega$-categorical} structure $\bB$. 
For an important subclass of $\omega$-categorical 
structures we have a tractability conjecture, too. 
%The conjecture says that if $\bA$ is a \emph{first-order reduct of a finitely bounded homogeneous structure}
The condition that supposedly characterises containment in P can be formulated in many non-trivially equivalent ways~\cite{BKOPP,BartoPinskerDichotomy,wonderland} and has been confirmed in numerous special cases, see for instance the articles~\cite{tcsps,BMPP16,posetCSP16,Phylo-Complexity,MMSNP,Bodirsky-Mottet} and the references therein. 

In the light of the recent advances in constraint satisfaction, both over finite and infinite domains, we revisit the RBCP and discuss the current state of the art. In particular, we observe that if $\fA$ has
a \emph{normal representation} (again, we follow the terminology of Hirsch~\cite{Hirsch}), then the network satisfaction problem for $\fA$ falls into the scope of the infinite-domain tractability conjecture. 
We also show that there is an algorithm that decides for a given finite relation algebra $\fA$ with a fully universal square representation 
whether $\fA$ has a normal representation. (In other words, there is an algorithm that decides for a given $\fA$ whether the class of atomic $\fA$-networks has the amalgamation property.) 
The scope of the tractability conjecture is larger, though. We describe an example of a finite relation algebra which has an $\omega$-categorical fully universal square representation 
(and a polynomial-time tractable network satisfaction problem) which is not normal, but which does fall into the scope of the conjecture.
%Full universality is meant here in the relation-algebraic sense: a representation of a relation algebra $\fA$ is \emph{fully universal} if it embeds all atomic $\fA$-networks. 

Whether the infinite-domain tractability conjecture
might contribute to the resolution of the RBCP in general remains open; we present several questions  in Section~\ref{sect:conclusion} whose answer would shed some light on this question. 
These questions concern the existence of $\omega$-categorical fully universal square representations and are of independent interest, and in my view they are central to the theory of representable finite relation algebras. 

% Question: suppose that A has a 
% universal omega-categorical representation.
% Does it then also have a sharp omega-cat
% universal representation? Candidate is mc core
% but it is not clear whether the existential definitions of the orbitals can be simulated with
% composition, so the structure could have
% more orbitals than its signature contains symbols. 

% Another question: is it true that a 
% universal omega-cat representation
% is finitely bounded? YES! 

% Do we care about sharpness of representations? 
% It does not seem to be the case. 

% BUT WHAT WE DO CARE: 
% Suppose that we have an omega-cat
% fully universal representation. 
% Do we have a fully universal representation
% which is a reduct of a homogeneous 
% structure with finite relational signature?
% This is really unclear. 

\section{Relation Algebras}
\label{sect:ra}
A \emph{proper relation algebra} is a set $B$ together with
a set $\mathcal R$ of binary relations over $B$ such that
\begin{enumerate}
\item $\Id := \{(x,x) \; | \; x \in B\} \in \mathcal R$;
\item If $R_1$ and $R_2$ are from $\mathcal R$, then $R_1 \vee R_2 := R_1 \cup R_2 \in \mathcal R$;
\item $1 := \bigcup_{R \in \mathcal R} R \in \mathcal R$;
\item $0 := \emptyset \in \mathcal R$;
\item If $R \in \mathcal R$, then $-R := 1 \setminus R \in \mathcal R$;
\item If $R \in \mathcal R$, then $R^{\smallsmile} := \{(x,y) \; | \; (y,x) \in R\} \in \mathcal R$;
\item If $R_1$ and $R_2$ are from $\mathcal R$, then $R_1 \circ R_2 \in \mathcal R$; where
$$ R_1 \circ R_2 := \{(x,z) \; | \; \exists y ((x,y) \in R_1 \wedge (y,z) \in R_2)\} \; .$$
\end{enumerate}

We want to point out that in this standard definition of proper relation algebras it is \emph{not} required that $1$ denotes $B^2$.
However, in most examples, 
$1$ indeed denotes $B^2$; in this case we say that the proper relation algebra is \emph{square}. 
The inclusion-wise 
minimal non-empty elements of $\mathcal R$ 
are called the \emph{basic relations} of the proper relation algebra.

\begin{example}[The Point Algebra]\label{expl:point-algebra}
Let $B={\mathbb Q}$ be the set of rational numbers,
and consider %set $\mathcal R$ containing the eight binary relations
$$\mathcal R = \{\emptyset,=,<,>,\leq,\geq,\neq,{\mathbb Q}^2\} \; .$$
Those relations form a proper relation algebra (with the basic relations $<,>,=$, and where $1$ denotes ${\mathbb Q}^2$) which is known under the name \emph{point algebra}.
\qed \end{example}

%When $\mathcal R$ is finite,
%every relation in $\mathcal R$
%can be written as a finite union of basic relations, and we abuse notation and sometimes write $R = \{B_1,\dots,B_k\}$ when $B_1, \dots,B_k$ are basic relations, $R \in \mathcal R$, and $R = B_1 \cup \dots \cup B_k$.
%Note that composition of basic
%relations determines the composition of all relations in the relation algebra, since $$R_1 \circ R_2 = \bigcup_{B_1 \in R_1, B_2 \in R_2}
%B_1 \circ B_2 \;.$$

The \emph{relation algebra associated to $(B,{\mathcal R})$} is the algebra $\fA$ with the domain 
$A := {\mathcal R}$ and the signature $\tau := \{\vee,-,0,1,\circ,^{\smallsmile},\Id\}$ obtained from $(B,{\mathcal R})$ in the obvious way. 
%We call $(D,{\mathcal R})$ a \emph{representation} of $\fA$. 
An \emph{abstract relation algebra} 
%(Definition~\ref{def:rel-algebra} below) 
is a $\tau$-algebra
that satisfies some of the laws that hold for the respective operators in a proper relation algebra. We do not need the precise definition of an abstract relation algebra in this
article since we deal exclusively with \emph{representable} relation algebras:
a \emph{representation}  
of an abstract relation algebra $\fA$ is a relational structure $\bB$ whose signature is $A$;
that is, the elements of the relation algebra are the relation symbols of $\bB$.
Each relation symbol $a \in A$ is associated
 to a binary relation $a^{\bB}$ 
over $B$ such that the set of relations of $\bB$ induces a proper relation algebra, and the map $a \mapsto a^{\bB}$ is an isomorphism with respect to the operations (and constants)
$\{\vee,-,0,1,\circ,^{\smallsmile},\Id\}$.
In this case, we also say that $\fA$ is the \emph{abstract relation algebra of $\bB$}.
An abstract relation algebra that has a representation is called
\emph{representable}. 
For $x,y \in A$, we write 
$x \leq y$ as a shortcut for the partial order defined by $x \vee y = y$. 
The minimal elements of $A \setminus \{0\}$
with respect to $\leq$ are called the \emph{atoms} of $\fA$. In every representation of $\fA$, the atoms denote the basic relations of the representation. 
We mention that 
there are abstract finite relation algebras
that are not representable~\cite{LyndonRelationAlgebras},
and that the question whether a finite relation
algebra is representable is undecidable~\cite{HirschHodkinsonRepresentability}.

\begin{example}\label{expl:a-pa}
The (abstract) point algebra
is a relation algebra with 8 elements and 3 atoms, $=$, $<$, and $>$, and can be described as follows.
The values of the composition operator for the atoms of the point algebra are shown in the table of Figure~\ref{fig:point-algebra}.
Note that this table determines the full composition table.
The inverse $(<)^{\smallsmile}$ of $<$ is $>$, and $\Id$ denotes $=$ which is its own inverse.
This fully determines the relation algebra.
\begin{figure}
\begin{center}
\begin{tabular}{|l||l|l|l|}
\hline
$\circ$ & $=$ & $<$ & $>$ \\
\hline \hline
$=$ & $=$ & $<$ & $>$ \\
\hline
$<$ & $<$ & $<$ & $1$ \\
\hline
$>$ & $>$ & $1$ & $>$ \\
\hline
\end{tabular}
\caption{The composition table for the basic relations in the point algebra.}
\label{fig:point-algebra}
\end{center}
\end{figure}
The proper relation algebra with domain $\mQ$ presented in Example~\ref{expl:point-algebra} is a representation of the point algebra. \qed
\end{example}

\section{The network satisfaction problem}
\label{sect:networks}
Let $\fA$ be a finite relation algebra with domain $A$. An \emph{${\fA}$-network} $N = (V;f)$ consists of a finite set of nodes $V$ and a function $f \colon V \times V \rightarrow A$. 

A network $N$ is called
\begin{itemize}
\item \emph{atomic}
if the image of $f$ only contains atoms of $\fA$ 
and if 
\begin{align}
f(a,c) \leq f(a,b) \circ f(b,c) \text{ for all $a,b,c \in V$}
\label{eq:atomic}
\end{align} 
(here we follow again the definitions in~\cite{Hirsch});
\item \emph{satisfiable in $\bB$}, for a representation $\bB$ of $\fA$, if there exists
 a map $s \colon V \to B$ (where $B$ denotes the domain of $\bB$) 
such that for all $x,y \in V$
$$(s(x),s(y)) \in f(x,y)^{\bB};$$
\item \emph{satisfiable} if $N$ is satisfiable in some representation $\bB$ of $\fA$. 
\end{itemize}

The \emph{(general) network satisfaction problem for a finite relation algebra ${\bf A}$} is the computational problem
to decide whether a given $\fA$-network 
is satisfiable. There are finite relation algebras
$\fA$ where 
this problem is undecidable~\cite{Hirsch-Undecidable}. 
A representation $\bB$ of $\fA$ 
is called 
\begin{itemize}
\item \emph{fully universal} if every atomic $\fA$-network is satisfiable in $\bB$; 
\item \emph{square} if its relations form a proper relation algebra that is square. 
\end{itemize}
The point algebra is an example of a relation algebra with a fully universal square representation.  
Note that if $\fA$ has a fully universal representation, then the network satisfaction problem for $\fA$ is decidable in NP: 
for a given network $(V,f)$, simply select for each $x \in V^2$ an atom $a \in A$ with $a \leq f(x)$,
replace $f(x)$ by $a$, and then exhaustively check condition~(\ref{eq:atomic}). 
Also note that a finite relation algebra
has a fully universal representation if and only if
the so-called path-consistency procedure decides satisfiability of atomic $\fA$-networks (see, e.g.,~\cite{Qualitative-Survey,HuangLR13}).

However, not all finite relation algebras have a fully universal representation. 
An example of a relation algebra
with 4 atoms which has a representation with seven 
elements but where path consistency of atomic networks 
does not imply consistency, called ${\bf B}_9$, has been given in~\cite{LiKRL08}. 
%(the example can in fact easily be simplified to one with 
%3 basic relations and a representation with five elements).
%; another example with 3 basic examples is due to~\cite{Maddux91}). 
A representation of ${\bf B}_9$ with domain
$\{0,1,\dots,6\}$ is given by the basic relations $\{R_0,R_1,R_2,R_3\}$ where 
$R_i = \{(x,y) : |x-y|=i \mod 7\}$, for $i \in \{0,1,2,3\}$. 
In fact, every representation of ${\bf B}_9$ is isomorphic
to this representation. Let $N$ be the network $(V,f)$ with $V = \{a,b,c,d\}$,
$f(a,b) = f(c,d) = R_3$, $f(a,d) = f(b,c) = R_2$, $f(a,c) = f(b,d) = R_1$, $f(i,i) = R_0$ for all $i \in V$, and $f(i,j) = f(j,i)$ for all $i,j \in V$. Then $N$ is atomic but not satisfiable. 

%Hence, ${\bf B}_9$ has an 
%$\omega$-categorical representation whose constraint satisfaction problem equals the general network satisfaction problem for ${\bf B}_9$. 

% TRY TO AVOID THESIS; MOREOVER,
% WOULD NEED NETWORKS TO BE PARTIAL MAPS
% WHICH IS SLIGHTLY NON-STANDARD
%Every finite relation algebra $\fA$ that has a representation also has a representation $\bB$ such that the general network satisfaction problem for $\fA$ and the network satisfaction problem for $\bB$ are one and the same problem.
%\begin{proposition}[Proposition~1.3.16 in~\cite{Bodirsky-HDR}]\label{prop:disj-union-networks}
%Every finite relation algebra $\bf A$ that has a representation also has a (countable) representation $\bB$ whose network satisfaction problem is the same problem as the general
% network satisfaction problem for $\fA$.
%\end{proposition}

\section{Constraint Satisfaction Problems}
Let $\bB$ be a structure with a (finite or infinite)
domain $B$ and a finite relational signature $\rho$. 
Then the \emph{constraint satisfaction problem for $\bB$} is the computational problem of deciding
whether a 
%finite conjunction of formulas of the form $R(x_1,\dots,x_k)$ for $R \in \rho$ is satisfiable in $\bB$. 
finite $\rho$-structure $\bE$ homomorphically maps to $\bB$. 
Note that if $\bB$ 
is a square representation of $\fA$, 
then the input $\bE$ can be viewed as an $\fA$-network $N$. 
The nodes
of $N$ are the elements of $\bE$. 
To define $f(x,y)$ for variables $x,y$ of the network,
let $a_1,\dots,a_k$ be a list of 
all elements $a \in A$ such that 
$(x,y) \in a^\bE$. Then define
$f(x,y) = (a_1 \wedge \cdots \wedge a_k)$;
if $k = 0$, then $f(x,y) = 1$. 
Observe that $\bE$ has a homomorphism to $\bB$ if and only if $N$ is satisfiable in $\bB$ (here we use the assumption that $\bB$ is a square representation). 

Conversely, when $N$ is an $\fA$-network, 
then we view $N$ 
as the $A$-structure $\bE$ whose domain are the nodes of $N$, 
and where $(x,y) \in r^{\bE}$ if and only if
$r = f(x,y)$. Again, $\bE$ has a homomorphism to $\bB$ if and only if $N$ is satisfiable in $\bB$. 

%Note that $N$ is atomic and satisfiable 
%if and only if the structure 
%$\bS_N$ is an induced substructure of a representation
%of $N$. 
%Conversely, every atomic $\fA$-network gives rise to an instance of $\Csp(\bB)$ in the obvious way. 

%the straightforward details can be found in~\cite{Bodirsky-HDR} or 
%~\cite{Qualitative-Survey}. 

%The following has been shown in~\cite{Bodirsky-HDR}.
%finite $\rho$-structure homomorphically maps to $\bA$. 

\begin{proposition}%[see~\cite{Bodirsky-HDR}]
\label{prop:square}
Let $\bB$ be a 
fully universal square representation of a finite relation algebra $\fA$. 
Then the network satisfaction problem for $\fA$
equals the constraint satisfaction problem for
$\bB$ (up to the translation between $\fA$-networks and finite $A$-structures presented above). 
\end{proposition}
\begin{proof}
We have to show that a network is satisfiable if and only if it has a homomorphism to $\bB$. Clearly,
if $N$ has a homomorphism to $\bB$ then it is satisfiable in $\bB$, and hence satisfiable. For the other direction, 
suppose that the $\fA$-network $N = (V,f)$ is satisfiable in some representation of $\fA$. Then there exists for each $x \in V^2$ an atomic 
$a \in A$ such that $a \leq f(x)$ and 
such that the network $N'$ obtained from
$N$ by replacing $f(x)$ by $a$
satisfies~(\ref{eq:atomic}); hence, $N'$ is atomic and satisfiable in $\bB$ since $\bB$ is fully universal. Hence, $N$ is satisfiable in $\bB$, too. 
\end{proof}

For general infinite structures $\bB$ a systematic understanding of the computational complexity of
 $\Csp(\bB)$ is a hopeless endeavour~\cite{BodirskyGrohe}. 
However, if $\bB$ is a \emph{first-order reduct of a finitely bounded homogeneous structure} (the definitions can be found below), then the universal-algebraic tractability conjecture for finite-domain CSPs can be generalised. 
This condition is sufficiently general 
so that it includes fully universal square representations of almost all the concrete finite relation algebras studied 
in the literature, and the condition also captures the class of finite-domain CSPs. 
As we will see, the concepts of 
\emph{finite boundedness} and \emph{homogeneity} are conditions that have already been studied in the relation algebra literature. 

\subsection{Finite boundedness}
Let $\rho$ be a relational signature, and let
 $\mathcal F$ be a set of $\rho$-structures. 
 Then $\Forb({\mathcal F})$ denotes the class
 of all finite $\rho$-structures $\bA$ such that no structure in ${\mathcal F}$ embeds into $\bA$.
For a $\rho$-structure $\bB$ 
 we write $\Age(\bB)$ for the class of all finite $\rho$-structures that embed into $\bB$. 
 We say that $\bB$ is \emph{finitely bounded}
if $\bB$ has a finite relational signature and  there exists a finite set of finite $\tau$-structures
${\mathcal F}$ such that $\Age(\bB) = \Forb(\mathcal F)$. A simple example of a finitely bounded structure is $({\mathbb Q};<)$. 
It is easy to see that the constraint satisfaction problem 
of a finitely bounded structure $\bB$ is in NP. 

\begin{proposition}\label{prop:fin-bound}
Let $\fA$ be a finite relation algebra
with a fully universal square representation 
$\bB$. Then $\bB$ is finitely bounded.
\end{proposition}
\begin{proof}[Proof sketch]
Besides some 
bounds of size at most two 
that make sure that 
the atomic relations partition $B^2$, 
it suffices to include appropriate three-element structures into ${\mathcal F}$ that can be read off from the composition table of $\fA$. 
\end{proof}

\subsection{Homogeneity}
A relational structure $\bB$ is \emph{homogeneous} 
(or \emph{ultra-homogeneous}~\cite{Hodges})
if every isomorphism between
finite 
substructures of $\bB$ can be extended to an automorphism of $\bB$. A simple example of a homogeneous structure is $({\mathbb Q};<)$.

A representation of a finite relation algebra
$\fA$ is called \emph{normal}
if it is square, fully universal, and homogeneous~\cite{Hirsch}. %It follows from Proposition~\ref{prop:fin-bound} that 
%if $\fA$ has a normal representation,
%then the network satisfaction problem for 
%$\bB$
The following is an immediate consequence of Proposition~\ref{prop:square} and
Proposition~\ref{prop:fin-bound}. 

\begin{corollary}\label{cor:normal}
Let $\fA$ be a finite relation algebra with a normal representation $\bB$. 
Then the network satisfaction problem for $\fA$ equals the constraint satisfaction problem for a finitely bounded homogeneous structure. 
\end{corollary}

%\subsection{\Fresse~Theory}
%\label{sect:fresse}
A versatile tool to construct homogeneous 
structures from classes of finite structures
is \emph{amalgamation} \`a la \Fresse. 
We present it for the special
case of \emph{relational structures}; 
this is all that is needed here. 
An \emph{embedding} of $\bA$ into $\bB$ is an isomorphism
between $\bA$ and a substructure of $\bB$.  
%In the following, let $\tau$ be a countable relational signature.
%The \emph{age} of a relational structure $\mathfrak A$ is the class of all finite structures 
%that embed into $\mathfrak A$.
An \emph{amalgamation diagram}
is a pair $(\bB_1,\bB_2)$
where $\bB_1,\bB_2$ are $\tau$-structures such that there exists a substructure $\bA$ of both $\bB_1$ and $\bB_2$ such that all common elements of $\bB_1$ and $\bB_2$ are elements of $\bA$. 
We say that $(\bB_1,\bB_2)$ is a \emph{2-point amalgamation diagram} if $|B_1 \setminus A| = |B_2 \setminus A| = 1$. 
%; note that in this case
%$\bA = \bB_1 \cap \bB_2$.  Then
%We call $\bB_1 \cup \bB_2$ the \emph{free amalgam} of $\bB_1,\bB_2$ over $\bA$. More
%generally,
A $\tau$-structure $\bC$ is an \emph{amalgam of $(\bB_1,\bB_2)$ over
  $\bA$} if for $i=1,2$ there are embeddings $e_i$ of $\bB_i$ to $\bC$ such that
$e_1(a)=e_2(a)$ for all $a \in A$.  
In the context of relation algebras $\fA$, 
the amalgamation property can also be formulated with atomic $\fA$-networks, in which case it has been called the \emph{patchwork property}~\cite{HuangLR13}; we stick with  the model-theoretic terminology here since it is older and well-established. 

\begin{definition}\label{def:amalgamation-prop}
An isomorphism-closed class $\mathcal C$ of
$\tau$-structures has the \emph{amalgamation property} if
every amalgamation diagram of structures in $\mathcal C$ has an amalgam in $\mathcal C$. 
A class of finite
$\tau$-structures that contains at most countably many non-isomorphic structures,
has the amalgamation property, and is closed under
taking induced substructures and isomorphisms
is called an \emph{amalgamation class}.
\end{definition}

Note that since we only look at relational structures here (and since we allow structures
to have an empty domain),
the amalgamation property of $\mathcal C$ implies the \emph{joint embedding property (JEP)} for $\mathcal C$, which says that for any two structures $\bB_1,\bB_2 \in \mathcal C$ there exists a structure $\bC \in \mathcal C$ that embeds both $\bB_1$ and $\bB_2$.

\begin{theorem}[\Fresse~\cite{OriginalFraisse,Fraisse}; see~\cite{Hodges}]\label{thm:fresse}
Let $\mathcal C$ be an amalgamation class.
Then there is a homogeneous and at 
most countable $\tau$-structure $\bC$ whose age equals $\mathcal C$.
The structure $\bC$ is unique up to isomorphism, and called
the \emph{\Fresse-limit} of $\mathcal C$.
\end{theorem}

The following is a well-known example of a finite relation algebra 
which has a fully universal square representation, but not a normal one.

\begin{example}\label{expl:branching-time}
The \emph{left linear point algebra} (see~\cite{HirschAlgebraicLogic,Duentsch}) 
is a relation algebra with four atoms, denoted by 
$=$, $<$, $>$, and $|$. Here we imagine that `$x < y$' signifies that
$x$ is \emph{earlier in time than $y$}. The idea
is that at every point in time the past is linearly ordered;
the future, however, is not yet determined and might branch into different worlds; incomparability of time points $x$ and $y$ is denoted by $x | y$. 
We might also think of $x < y$ as 
 \emph{$x$ is to the left of $y$} if we draw points in the plane, and this motivates the name \emph{left linear point algebra}.
The composition operator on those four basic relations is given in Figure~\ref{fig:left-linear}. The inverse $(<)^{\smallsmile}$ of $<$ is $>$, $\Id$ denotes $=$,
and $|$ is its own inverse, and the relation algebra is uniquely
given by this data. 
\begin{figure}
\begin{center}
\begin{tabular}{|c||c|c|c|c|}
\hline
$\circ$ & $=$ & $<$ & $>$ & $|$ \\
\hline \hline
$=$ & $=$ & $<$ & $>$ & $|$ \\
\hline
$<$ & $<$ & $<$ & $\{<,=,>\}$ & $\{<,|\}$ \\
\hline
$>$ & $>$ & $1$ & $>$ & $|$ \\
\hline
$|$ & $|$ & $|$ & $\{>,|\}$ & 1 \\ 
\hline
\end{tabular}
\end{center}
\caption{The composition table for the basic relations in the left linear point algebra.}
\label{fig:left-linear}
\end{figure}
It is well known (for details, see~\cite{Bodirsky}) that the left linear point algebra has a fully universal square representation. 
On the other hand, the %following 
networks 
drawn in Figure~\ref{fig:amalgam-fails} 
show the failure of amalgamation. 
\end{example}

An algorithm to test whether a finite relation algebra has a normal representation can be found in Section~\ref{sect:testing-normal}.

\begin{figure}
  \begin{center}
  \includegraphics[scale=0.5]{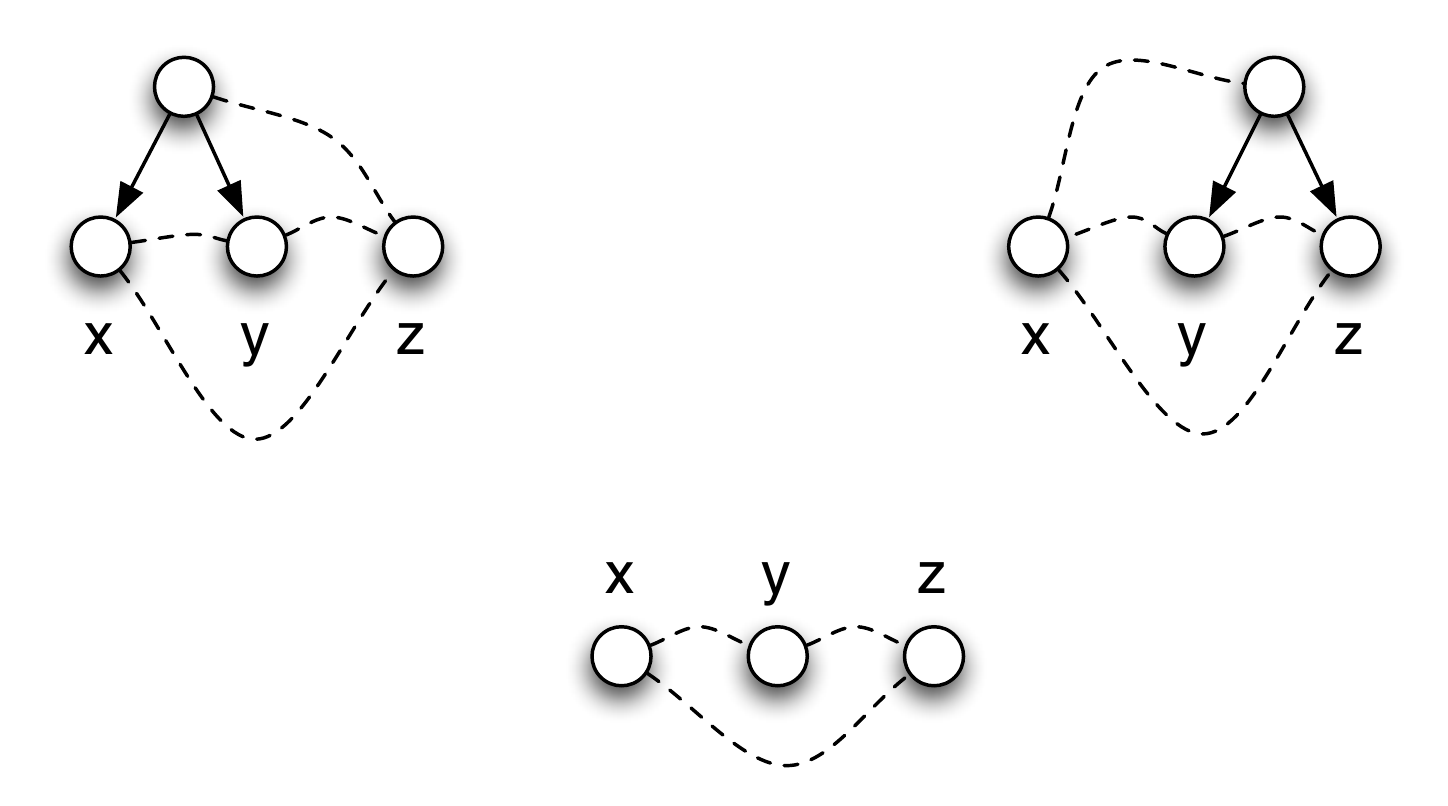}
  \end{center}
\caption{Example showing that atomic networks for the left linear point algebra do not have the amalgamation property. A directed edge from $x$ to $y$ signifies $x < y$, and a dashed edge between $x$ and $y$ signifies $x | y$. 
}
  \label{fig:amalgam-fails}
\end{figure}

\subsection{The infinite-domain dichotomy conjecture}
The infinite-domain dichotomy conjecture applies to a class which is larger than the class of homogeneous finitely bounded structures. To introduce this class we need the concept of \emph{first-order reducts}. 

Suppose that two relational structures 
$\bA$ and $\bB$ have the same domain, 
that the signature of a structure $\bA$ is a subset of the signature of $\bB$,  and that $R^\bA = R^\bB$ for all common relation symbols $R$. Then we call 
$\bA$ a \emph{reduct of $\bB$}, and
$\bB$ an \emph{expansion of $\bA$}. 
In other words, $\bA$ is obtained from $\bB$ by dropping
some of the relations. 
A \emph{first-order reduct of $\bB$} is
a reduct of the expansion of $\bB$ by all relations that are first-order definable in $\bB$. The CSP for a first-order reduct of a finitely bounded homogeneous structure is in NP (see~\cite{Bodirsky-HDR}). 
An example of a structure which is not homogeneous, but a reduct of finitely bounded homogeneous structure is the representation of the left-linear point algebra (Example~\ref{expl:branching-time}) given in~\cite{Bodirsky}. 

\begin{conjecture}[Infinite-domain dichotomy conjecture]
Let $\bB$ be a first-order reduct of a finitely bounded homogeneous structure. Then $\Csp(\bB)$ is either in P or NP-complete. 
\end{conjecture} 

%The representation of the left-linear point algebra given 
%TODO 
%The fact that there exists an $\omega$-categorical
%structure $\bB$ such that the CSP for $\bB$ 
%equals the general network satisfaction problem of the 
%branching time point algebra 
%has already been observed in~\cite{BodirskyNesetrilJLC}. 

Hence, the infinite-domain dichotomy conjecture implies the RBCP for finite relation algebras with a normal representation. In Section~\ref{sect:tractability} we will see a more specific conjecture that characterises the NP-complete cases and
the cases that are in P. 
\section{The Infinite-Domain Tractability Conjecture}
\label{sect:tractability}
To state the infinite-domain tractability conjecture, we need a couple of concepts that are most naturally introduced for the class of all $\omega$-categorical structures. 
A theory is called \emph{$\omega$-categorical} if all its countably infinite models are isomorphic.
A structure is called \emph{$\omega$-categorical} if its first-order theory is $\omega$-categorical. 
Note that finite structures are $\omega$-categorical since their first-order theories do not have countably infinite models. 
Homogeneous structures $\bB$ with finite relational signature are $\omega$-categorical. This follows from a very useful  characterisation of $\omega$-categoricity given by Engeler, Svenonius, and Ryll-Nardzewski (Theorem~\ref{thm:ryll}).
%An \emph{isomorphism} is a bijective homomorphism $h$ 
%such that the inverse mapping $h^{-1} \colon B \rightarrow A$ that sends $h(x)$ to $x$ 
%whose inverse is a homomorphism, too.
%An \emph{automorphism} of $\bB$ is an isomorphism
%between $\bB$ and itself. 
The set of all automorphisms of $\bB$ 
is denoted by $\Aut(\bB)$. The \emph{orbit} of a $k$-tuple
$(t_1,\dots,t_n)$ under $\Aut(\bB)$ is the set
$\{(a(t_1),\dots,a(t_n)) \mid a \in \Aut(\bB) \}$. Orbits of pairs (i.e., $2$-tuples)
are also called \emph{orbitals}. 

\begin{theorem}[see~\cite{Hodges}]\label{thm:ryll}
A countable structure $\bB$ is $\omega$-categorical if and only if
$\Aut(\bB)$ has only finitely many orbits of $n$-tuples, for all $n \geq 1$. 
\end{theorem}

The following is an easy consequence of Theorem~\ref{thm:ryll}.

\begin{proposition}
First-order reducts of $\omega$-categorical structures are $\omega$-categorical. 
\end{proposition}

First-order reducts of homogeneous structures, on the other hand, need not be homogeneous. 
An example of an $\omega$-categorical structure which is not homogeneous is
the $\omega$-categorical representation of the left linear point algebra given in~\cite{Bodirsky} (see Example~\ref{expl:branching-time}). 
Note that every $\omega$-categorical structure $\bB$, and more generally every structure with finitely many orbitals,  gives rise to a finite relation algebra, namely the relation algebra associated to the unions of orbitals of $\bB$ (see~\cite{Qualitative-Survey});
we refer to this relation algebra as the
\emph{orbital relation algebra} of $\bB$. 

We first present a condition that implies that
an $\omega$-categorical structure has
an NP-hard constraint satisfaction problem (Section~\ref{sect:original-conj}). 
The tractability conjecture says that 
every reduct of a finitely bounded homogeneous structure that does not satisfy this condition is NP-complete. 
We then present an equivalent characterisation of the condition due to Barto and Pinsker (Section~\ref{sect:BP}), 
and then yet another condition due to Barto, Opr\v{s}al, and Pinsker, which was later shown to be equivalent (Section~\ref{sect:wonderland}). 
%which was shown to be equivalent in~\cite{}. 

\subsection{The original formulation of the conjecture}
\label{sect:original-conj}
Let $\bB$ be an $\omega$-categorical structure. Then $\bB$ 
is called  
\begin{itemize}
\item a \emph{core} if all endomorphisms
of $\bB$ 
(i.e., homomorphisms from $\bB$ to $\bB$)
are embeddings (i.e., are injective and also preserve the complement of each relation). 
\item \emph{model complete} if all
self-embeddings of $\bB$ are elementary, i.e., preserve all first-order formulas. 
\end{itemize}

Clearly, if $\bB$ is a representation of 
a finite relation algebra $\fA$, then
$\bB$ is a core. However, not all representations of finite relation algebras are model complete. A simple example is the orbital relation algebra of the structure $({\mathbb Q}_0^+;<)$ where 
${\mathbb Q}_0^+$ denotes the non-negative rationals: its representation 
with domain ${\mathbb Q}_0^+$ has self-embeddings that do not preserve the orbital $\{(0,0)\}$. 

Let $\tau$ be a relational signature. 
A $\tau$-formula is called \emph{primitive positive} if it is of the form 
$\exists x_1,\dots,x_n (\psi_1 \wedge \cdots \wedge \psi_m)$ where
$\psi_i$ is of the form $y_1 = y_2$ or of the form $R(y_1,\dots,y_k)$ for $R \in \tau$ of arity $k$. The variables $y_1,\dots,y_k$ can be free or from $x_1,\dots,x_n$. Clearly, primitive positive formulas are preserved by homomorphisms. 

\begin{theorem}[\cite{Cores-journal,BodHilsMartin}]
\label{thm:mc-core}
Every $\omega$-categorical structure $\bB$ is homomorphically equivalent to a model-complete core $\bC$, which is unique up to isomorphism, and again $\omega$-categorical. All orbits of $k$-tuples are primitive positive definable in $\bC$. 
\end{theorem}

The (up to isomorphism unique) structure $\bC$
from Theorem~\ref{thm:mc-core} is called \emph{the model-complete core of $\bB$.}
Let $\bB$ and $\bA$ be structures, let
$D \subseteq B^n$, and let $I \colon D \to A$ be a surjection. Then $I$ is called a \emph{primitive positive interpretation}
if the pre-image under $I$ of $A$, of the equality relation $=_A$ on $A$, and of all relations of $\bA$ is primitive positive definable in $\bA$. 
In this case we also say that \emph{$\bB$ interprets $\bA$ primitively positively}. 
The complete graph with three vertices (but without loops) is denoted by $K_3$. 

\begin{theorem}[\cite{BodirskySurvey}]
\label{thm:hardness}
Let $\bB$ be an $\omega$-categorical structure. If the model-complete core of $\bB$ has an expansion by finitely many constants so that the resulting structure interprets $K_3$ primitively positively, then $\Csp(\bB)$ is NP-hard. 
\end{theorem}
% (equivalently: interprets all finite structures)

We can now state the infinite-domain tractability conjecture. 

\begin{conjecture}
Let $\bB$ be a first-order reduct of a
finitely bounded homogeneous structure.
If $\bB$ does not satisfy the condition from Theorem~\ref{thm:hardness} then 
$\Csp(\bB)$ is in P. 
\end{conjecture}
This conjecture has been verified in numerous special cases (see, for instance, the articles~\cite{tcsps,BMPP16,posetCSP16,Phylo-Complexity,MMSNP,Bodirsky-Mottet}), including the class of finite-structures~\cite{BulatovFVConjecture,ZhukFVConjecture}. 

\subsection{The theorem of Barto and Pinsker}
\label{sect:BP}
The tractability conjecture has a fundamentally different, but equivalent formulation: instead of the \emph{non-existence} of a hardness-condition, we require the \emph{existence} of a polymorphism satisfying a certain identity; the concept of polymorphisms
is fundamental to the resolution of the Feder-Vardi conjecture in both~\cite{BulatovFVConjecture} and
\cite{ZhukFVConjecture}. 

\begin{definition}
A \emph{polymorphism} of a structure $\bB$ is a homomorphism from $\bB^k$ to $\bB$, for some $k \in {\mathbb N}$. 
We write $\Pol(\bB)$ for the set of all polymorphisms of $\bB$. 
\end{definition}

An operation $f \colon B^6 \to B$ is called 
\begin{itemize}
\item \emph{Siggers} if it satisfies
$$f(x,y,x,z,y,z) = f(z,z,y,y,x,x)$$
for all $x,y,z \in B$; 
\item \emph{pseudo-Siggers modulo $e_1,e_2 \colon B \to B$} 
if $$e_1(f(x,y,x,z,y,z)) = e_2(f(z,z,y,y,x,x))$$
for all $x,y,z \in B$. 
\end{itemize}

\begin{theorem}[\cite{BartoPinskerDichotomy}]
Let $\bB$ be an $\omega$-categorical model-complete core. Then either
\begin{itemize}
\item $\bB$ can be expanded by finitely many constants so that the resulting structure interprets 
%all finite structures 
$K_3$ 
primitively positively, or
\item $\bB$ has a pseudo-Siggers polymorphism modulo endomorphisms of $\bB$. 
\end{itemize}
\end{theorem}

\subsection{The wonderland conjecture}
\label{sect:wonderland}
A weaker condition that implies that 
an $\omega$-categorical structure has
an NP-hard CSP has been presented in
~\cite{wonderland}. For reducts of homogeneous structures with finite signature, however, the two conditions are equivalent~\cite{BKOPP}. 
Hence, we obtain yet another different but equivalent formulation of the tractability conjecture. The advantage of the new formulation is that it does not require that
the structure is a model-complete core. 

Let $\bB$ be a countable structure. 
A map $\mu \colon \Pol(\bB) \to \Pol(\bA)$ is called \emph{minor-preserving} if
for every $f \in \Pol(\bB)$ of arity $k$
and all $k$-ary projections $\pi_1,\dots,\pi_k$ we have $\mu(f) \circ (\pi_1,\dots,\pi_k) = \mu(f \circ (\pi_1,\dots,\pi_k))$ where $\circ$ denotes composition of functions. 
The set $\Pol(\bB)$ is equipped with a natural complete %left non-expansive and projection right invariant 
ultrametric $d$ (see, e.g.,~\cite{BodirskySchneiderTopological}). To define $d$, suppose that $B = {\mathbb N}$. For $f,g \in \Pol(\bB)$ we define
$d(f,g) = 1$ if $f$ and $g$ have different arity; otherwise, if both $f,g$ have arity $k \in {\mathbb N}$, then $$d(f,g) :=  2^{- \min\{n \in {\mathbb N} \mid \exists s \in \{1,\dots,n\}^k: f(s) \neq g(s)\}}.$$

\begin{theorem}[of~\cite{wonderland}]
\label{thm:wonderland}
Let $\bB$ be $\omega$-categorical. 
Suppose that $\Pol(\bB)$ has a uniformly continuous minor-preserving map to 
$\Pol(K_3)$. Then $\Csp(\bB)$ is NP-complete. 
\end{theorem}
We mention that there are $\omega$-categorical structures where the condition from Theorem~\ref{thm:wonderland} applies, but not the condition from Theorem~\ref{thm:hardness}~\cite{BKOPP}. 
%However, 
%for reducts of finitely bounded homogeneous structures the two conditions are equivalent. 

\begin{theorem}[of~\cite{BKOPP}]
If $\bB$ is a reduct of a homogeneous structure with finite relational signature, then the conditions given in Theorem~\ref{thm:hardness} and
in Theorem~\ref{thm:wonderland} are equivalent. 
\end{theorem}

\section{Testing the Existence of Normal Representations}
\label{sect:testing-normal}
In this section we present an algorithm that tests whether a given finite relation algebra has a normal representation. This follows from a model-theoretic result that seems to be folklore, namely that testing the amalgamation property for a class of structures that has the JEP and a signature of maximal arity two which is given by a finite set of forbidden substructures is decidable. We are not aware of a proof of this in the literature. 

\begin{theorem}\label{thm:normal}
There is an algorithm that decides for a given finite relation algebra $\fA$ 
which has a fully universal square representation whether 
$\fA$ also has a normal representation. 
\end{theorem}

\begin{proof}
First observe that the class $\mathcal C$ 
of all atomic 
$\fA$-networks, viewed as $A$-structures, has the JEP: if $N_1$ and $N_2$ are atomic networks, then they are satisfiable in $\bB$ since $\bB$ is fully universal, 
and hence embed into $\bB$ when viewed as structures.  
Since $\bB$ is square the substructure of $\bB$ 
induced by the union of the images of $N_1$ and $N_2$ is an atomic network, too,
and it embeds $N_1$ and $N_2$. 

Let $k$ be the number of atoms of $\fA$. 
It clearly suffices to show the following claim, since the condition given there can be effectively checked exhaustively. 

{\bf Claim.} $\mathcal C$ has the AP if and only if all 2-point amalgamation diagrams of size at most $k+2$ amalgamate. 

So suppose that $D = (\bB_1,\bB_2)$
%(\bB_0,\bB_1,\bB_2,e_1,e_2)$ 
is an amalgamation diagram without amalgam. Let $\bB_1'$ be a maximal substructure of $\bB_1$ that contains $B_1 \cap B_2$ such that $(\bB_1',\bB_2)$ has an amalgam. 
Let $\bB_2'$ be a maximal substructure
of $\bB_2$ that contains $B_1 \cap B_2$ such that $(\bB_1',\bB_2')$ has an amalgam. Then $B_i \neq B_i'$ for some
$i \in \{1,2\}$; let $\bC_1$ be a substructure of $\bB_i$ that extends $\bB_i'$ by one element, and let $\bC_2 := \bB_{3-i}'$. 
Then $(\bC_1,\bC_2)$ is a 2-point amalgamation diagram without an amalgam. Let $C_0 := C_1 \cap C_2$.
Let $C_1 \setminus C_0 = \{p\}$ and $C_2 \setminus C_0 = \{q\}$. 
For each $a \in A$ there exists an element $r_a \in C_0$ such that the network $(\{r,p,q\},f)$
with $f(p,q) = a$, $f(p,r) = f^{\bB_1}(p,r)$,
$f(r,q) = f^{\bB_2}(r,q)$ fails the atomicity property~(\ref{eq:atomic}). 
Let $\bC_1'$ be the substructure of $\bC_1$ induced by $\{p\} \cup \{r_a \mid a \in A\}$ and $\bA_1'$ be the substructure of $\bC_2$ induced by $\{q\} \cup \{r_a \mid a \in A\}$. Then the amalgamation diagram 
$(\bC_1',\bC_2')$ has no amalgam, and
has size at most $k+2$. 
\end{proof}

\section{Conclusion and Open Problems}
\label{sect:conclusion}
Hirsch's Really Big Complexity Problem (RBCP) for finite relation algebras remains really big. 
However, the network satisfaction problem of every finite relation algebra 
%that have a fully universal representation 
known to the author can be formulated as the CSP of a structure that falls into the scope of the infinite-domain tractability conjecture. Most of the classical examples even have a normal representation, and therefore the RBCP for those is implied by the infinite-domain tractability conjecture (Corollary~\ref{cor:normal}). 
%Relation algebras with a normal representation 
%provably fall into the scope of the infinite-domain tractability conjecture.
We presented an algorithm that tests whether a given finite relation algebra has a normal representation.

\begin{figure}
  \begin{center}
  \includegraphics[scale=0.5]{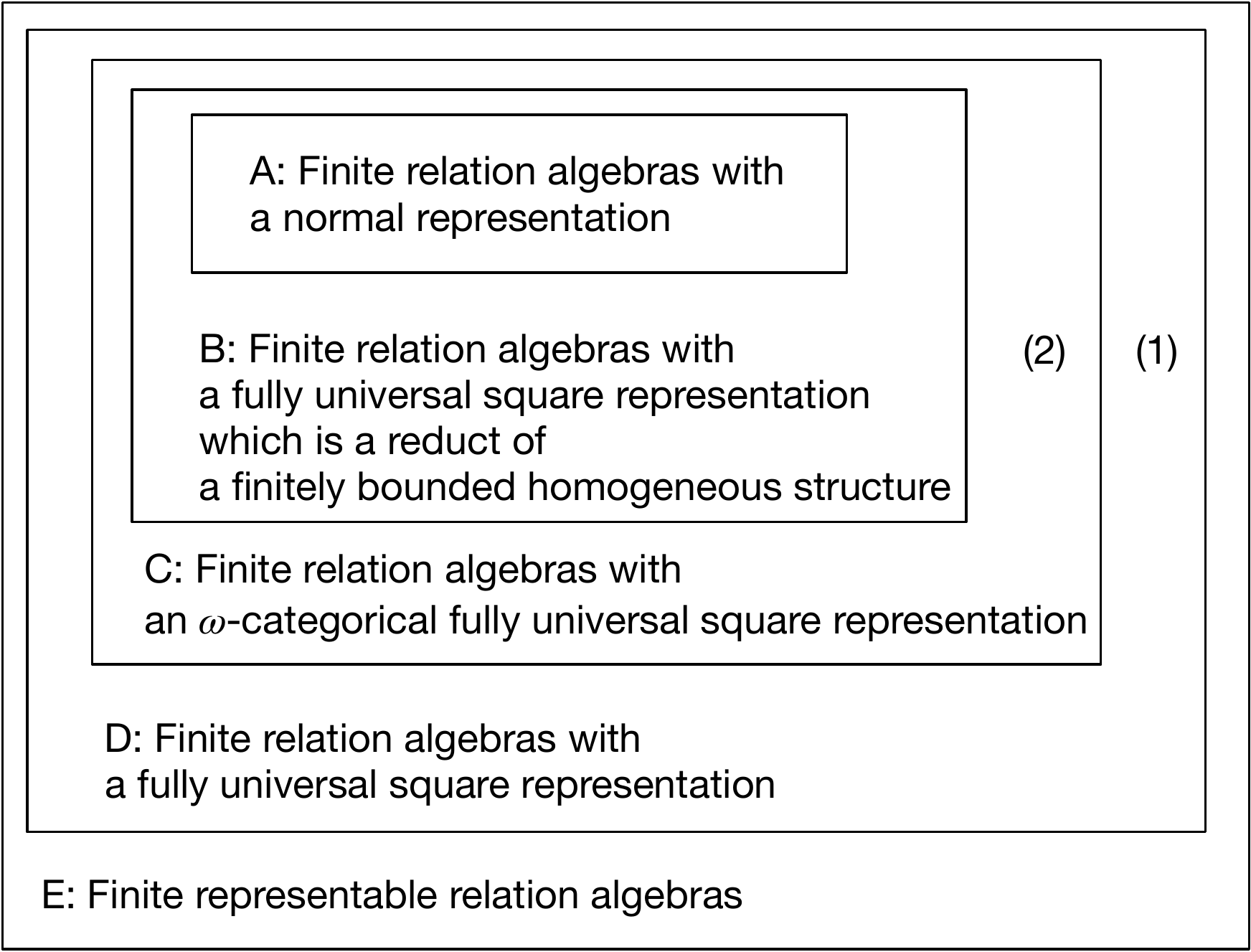}
  \end{center}
\caption{Subclasses of finite representable relation algebras. Membership of relation algebras from D to the innermost box A is decidable (Theorem~\ref{thm:normal}). 
Example~\ref{expl:branching-time} separates Box A and Box B. 
The finite relation algebra from~\cite{Hirsch-Undecidable} separates Box D and E. 
Box B falls into the scope of the infinite-domain tractability conjecture. 
Boxes C and D might also fall into the scope of this conjecture (see Problem (1) and Problem (2)). 
%Even Box E might fall into the scope of the conjecture (see Problem (3)). 
}
  \label{fig:5-Levels}
\end{figure}

To better understand the RBCP in general, or at least for finite relation algebras with fully universal square representation, we need a better understanding of representations of finite relation algebras with good model-theoretic properties. We mention some concrete open questions;
also see  Figure~\ref{fig:5-Levels}.

\begin{enumerate}
% Sharp means: atoms are orbitals. 
% May be not so important in our context. 
% Moreover: when we work with mc cores,
% get ep definitions of the orbits that might
% not be obtainable by composition.
%
%\item Is there a representable finite relation algebra without a sharp representation? 
%\item Is there a finite relation algebra with a sharp and universal representation, but without an $\omega$-categorical sharp representation?
%\item Is there a finite relation algebra with a homogeneous representation that does not satisfy the condition of Huang, Li, and Renz (path consistency decides satisfiability of atomic networks)? Yes, see earlier. 
%\item Is there a finite relation algebra $\fA$ with a homogeneous universal representation but without
%a universal homogeneous and sharp representation? 
\item Is there a finite relation algebra with a fully
universal square representation, but without an $\omega$-categorical 
fully universal square representation? 
\item Is there a finite relation algebra with an $\omega$-categorical fully universal square representation 
but without a fully universal square representation which is not a first-order reduct of a finitely bounded homogeneous structure? 
%\item Is there a finite relation algebra with a fully universal representation, but without fully universal representation having a finite-signature first-order homogenisation? 
\item Find a finite relation
algebra $\fA$ such that there is no 
$\omega$-categorical
structure $\bB$ such that the general network satisfaction problem for $\fA$ equals the constraint satisfaction problem for $\bB$. (Note that we do not insist on $\bB$ being a representation of $\fA$.) 
\item Find a finite relation algebra $\fA$ with an $\omega$-categorical fully universal square representation which is not the orbital relation algebra of an $\omega$-categorical structure. 
\end{enumerate}

% (see the concrete questions in Section~\ref{sect:problems}). 
%Hence, if some of the questions in 
%Section~\ref{sect:problems} have the answer 
%\emph{no}, then pictorially speaking this decreases 

\bibliographystyle{alpha}
\bibliography{local.bib}

\newcommand{\etalchar}[1]{$^{#1}$}
\def\cprime{$'$} \def\cprime{$'$}
\begin{thebibliography}{BKO{\etalchar{+}}17}

\bibitem[BG08]{BodirskyGrohe}
Manuel Bodirsky and Martin Grohe.
\newblock Non-dichotomies in constraint satisfaction complexity.
\newblock In Luca Aceto, Ivan Damgard, Leslie~Ann Goldberg, Magn\'us~M.
  Halld\'orsson, Anna Ing\'olfsd\'ottir, and Igor Walukiewicz, editors, {\em
  Proceedings of the International Colloquium on Automata, Languages and
  Programming (ICALP)}, Lecture Notes in Computer Science, pages 184 --196.
  Springer Verlag, July 2008.

\bibitem[BHM10]{BodHilsMartin}
Manuel Bodirsky, Martin Hils, and Barnaby Martin.
\newblock On the scope of the universal-algebraic approach to constraint
  satisfaction.
\newblock In {\em Proceedings of the Annual Symposium on Logic in Computer
  Science (LICS)}, pages 90--99. IEEE Computer Society, July 2010.

\bibitem[BJ17]{Qualitative-Survey}
Manuel Bodirsky and Peter Jonsson.
\newblock A model-theoretic view on qualitative constraint reasoning.
\newblock {\em Journal of Artificial Intelligence Research}, 58:339--385, 2017.

\bibitem[BJP17]{Phylo-Complexity}
Manuel Bodirsky, Peter Jonsson, and Trung~Van Pham.
\newblock {The Complexity of Phylogeny Constraint Satisfaction Problems}.
\newblock {\em ACM Transactions on Computational Logic (TOCL)}, 18(3), 2017.
\newblock An extended abstract appeared in the conference STACS 2016.

\bibitem[BK08]{tcsps}
Manuel Bodirsky and Jan K\'ara.
\newblock The complexity of temporal constraint satisfaction problems.
\newblock In Cynthia Dwork, editor, {\em Proceedings of the Annual Symposium on
  Theory of Computing (STOC)}, pages 29--38. ACM, May 2008.

\bibitem[BK09]{BoundedWidth}
Libor Barto and Marcin Kozik.
\newblock Constraint satisfaction problems of bounded width.
\newblock In {\em Proceedings of the Annual Symposium on Foundations of
  Computer Science (FOCS)}, pages 595--603, 2009.

\bibitem[BKJ05]{JBK}
Andrei~A. Bulatov, Andrei~A. Krokhin, and Peter~G. Jeavons.
\newblock Classifying the complexity of constraints using finite algebras.
\newblock {\em SIAM Journal on Computing}, 34:720--742, 2005.

\bibitem[BKO{\etalchar{+}}17]{BKOPP}
Libor Barto, Michael Kompatscher, Miroslav Ol\v{s}\'{a}k, Michael Pinsker, and
  Trung~Van Pham.
\newblock The equivalence of two dichotomy conjectures for infinite domain
  constraint satisfaction problems.
\newblock In {\em Proceedings of the 32nd Annual {ACM/IEEE} Symposium on Logic
  in Computer Science -- LICS'17}, 2017.
\newblock Preprint arXiv:1612.07551.

\bibitem[BM16]{Bodirsky-Mottet}
Manuel Bodirsky and Antoine Mottet.
\newblock Reducts of finitely bounded homogeneous structures, and lifting
  tractability from finite-domain constraint satisfaction.
\newblock In {\em Proceedings of the 31th Annual IEEE Symposium on Logic in
  Computer Science -- LICS'16}, pages 623--632, 2016.
\newblock Preprint available at ArXiv:1601.04520.

\bibitem[BMM18]{MMSNP}
Manuel Bodirsky, Florent Madelaine, and Antoine Mottet.
\newblock {A universal-algebraic proof of the complexity dichotomy for Monotone
  Monadic SNP}.
\newblock In {\em Proceedings of the Symposium on Logic in Computer Science --
  LICS'18}, 2018.
\newblock Preprint available under ArXiv:1802.03255.

\bibitem[BMPP16]{BMPP16}
Manuel Bodirsky, Barnaby Martin, Michael Pinsker, and Andr{\'{a}}s
  Pongr{\'{a}}cz.
\newblock Constraint satisfaction problems for reducts of homogeneous graphs.
\newblock In {\em 43rd International Colloquium on Automata, Languages, and
  Programming, {ICALP} 2016, July 11-15, 2016, Rome, Italy}, pages
  119:1--119:14, 2016.

\bibitem[Bod04]{Bodirsky}
Manuel Bodirsky.
\newblock Constraint satisfaction with infinite domains.
\newblock Dissertation, Humboldt-Universit\"at zu Berlin, 2004.

\bibitem[Bod07]{Cores-journal}
Manuel Bodirsky.
\newblock Cores of countably categorical structures.
\newblock {\em Logical Methods in Computer Science}, 3(1):1--16, 2007.

\bibitem[Bod08]{BodirskySurvey}
Manuel Bodirsky.
\newblock Constraint satisfaction problems with infinite templates.
\newblock In Heribert Vollmer, editor, {\em Complexity of Constraints (a
  collection of survey articles)}, volume 5250 of {\em Lecture Notes in
  Computer Science}, pages 196--228. Springer, 2008.

\bibitem[Bod12]{Bodirsky-HDR}
Manuel Bodirsky.
\newblock Complexity classification in infinite-domain constraint satisfaction.
\newblock M\'emoire d'habilitation \`a diriger des recherches, Universit\'{e}
  Diderot -- Paris 7. Available at arXiv:1201.0856, 2012.

\bibitem[BOP17]{wonderland}
Libor Barto, Jakub Opr\v{s}al, and Michael Pinsker.
\newblock The wonderland of reflections.
\newblock {\em Israel Journal of Mathematics}, 2017.
\newblock To appear. Preprint arXiv:1510.04521.

\bibitem[BP16]{BartoPinskerDichotomy}
Libor Barto and Michael Pinsker.
\newblock The algebraic dichotomy conjecture for infinite domain constraint
  satisfaction problems.
\newblock In {\em Proceedings of the 31th Annual IEEE Symposium on Logic in
  Computer Science -- LICS'16}, pages 615--622, 2016.
\newblock Preprint arXiv:1602.04353.

\bibitem[BS16]{BodirskySchneiderTopological}
Manuel Bodirsky and Friedrich~Martin Schneider.
\newblock A topological characterisation of endomorphism monoids of countable
  structures.
\newblock {\em Algebra Universalis}, 77(3):251--269, 2016.
\newblock Preprint available at arXiv:1508.07404.

\bibitem[Bul17]{BulatovFVConjecture}
Andrei~A. Bulatov.
\newblock A dichotomy theorem for nonuniform {CSP}s.
\newblock In {\em 58th {IEEE} Annual Symposium on Foundations of Computer
  Science, {FOCS} 2017, Berkeley, CA, USA, October 15-17, 2017}, pages
  319--330, 2017.

\bibitem[D{\"{u}}n05]{Duentsch}
Ivo D{\"{u}}ntsch.
\newblock Relation algebras and their application in temporal and spatial
  reasoning.
\newblock {\em Artificial Intelligence Review}, 23:315--357, 2005.

\bibitem[Fra54]{OriginalFraisse}
Roland Fra{\"\i}ss\'e.
\newblock Sur l'extension aux relations de quelques propri\'et\'es des ordres.
\newblock {\em Annales Scientifiques de l'\'Ecole Normale Sup\'erieure},
  71:363--388, 1954.

\bibitem[Fra86]{Fraisse}
Roland Fra{\"\i}ss\'e.
\newblock {\em Theory of Relations}.
\newblock Elsevier Science Ltd, North-Holland, 1986.

\bibitem[HH01]{HirschHodkinsonRepresentability}
R.~Hirsch and I.~Hodkinson.
\newblock Representability is not decidable for finite relation algebras.
\newblock {\em Transactions of the American Mathematical Society},
  353(4):1387--1401), 2001.

\bibitem[Hir96]{Hirsch}
Robin Hirsch.
\newblock Relation algebras of intervals.
\newblock {\em Artificial Intelligence Journal}, 83:1--29, 1996.

\bibitem[Hir97]{HirschAlgebraicLogic}
Robin Hirsch.
\newblock Expressive power and complexity in algebraic logic.
\newblock {\em Journal of Logic and Computation}, 7(3):309 -- 351, 1997.

\bibitem[Hir99]{Hirsch-Undecidable}
Robin Hirsch.
\newblock A finite relation algebra with undecidable network satisfaction
  problem.
\newblock {\em Logic Journal of the {IGPL}}, 7(4):547--554, 1999.

\bibitem[HLR13]{HuangLR13}
Jinbo Huang, Jason~Jingshi Li, and Jochen Renz.
\newblock Decomposition and tractability in qualitative spatial and temporal
  reasoning.
\newblock {\em Artif. Intell.}, 195:140--164, 2013.

\bibitem[Hod97]{Hodges}
Wilfrid Hodges.
\newblock {\em A shorter model theory}.
\newblock Cambridge University Press, Cambridge, 1997.

\bibitem[KP17]{posetCSP16}
Michael Kompatscher and Trung~Van Pham.
\newblock {A Complexity Dichotomy for Poset Constraint Satisfaction}.
\newblock In {\em 34th Symposium on Theoretical Aspects of Computer Science
  (STACS 2017)}, volume~66 of {\em Leibniz International Proceedings in
  Informatics (LIPIcs)}, pages 47:1--47:12, 2017.

\bibitem[LKRL08]{LiKRL08}
Jason~Jingshi Li, Tomasz Kowalski, Jochen Renz, and Sanjiang Li.
\newblock Combining binary constraint networks in qualitative reasoning.
\newblock In {\em {ECAI} 2008 - 18th European Conference on Artificial
  Intelligence, Patras, Greece, July 21-25, 2008, Proceedings}, pages 515--519,
  2008.

\bibitem[Lyn50]{LyndonRelationAlgebras}
R.~Lyndon.
\newblock The representation of relational algebras.
\newblock {\em Annals of Mathematics}, 51(3):707--729, 1950.

\bibitem[Zhu17]{ZhukFVConjecture}
Dmitriy Zhuk.
\newblock A proof of {CSP} dichotomy conjecture.
\newblock In {\em 58th {IEEE} Annual Symposium on Foundations of Computer
  Science, {FOCS} 2017, Berkeley, CA, USA, October 15-17, 2017}, pages
  331--342, 2017.

\end{thebibliography}

\end{document}